\documentclass[11pt]{article}
\usepackage{amsmath,amsfonts,amsthm,a4,amssymb}
\usepackage{enumitem,xspace,url,dsfont,engord}

\setenumerate{topsep=3pt,parsep=-2pt,leftmargin=*}
\setenumerate[3]{label=(\roman*)}
\setenumerate[2]{label=(\alph*)}
\setenumerate[1]{label=(\roman*)}

\newcommand{\set}[1]{\left\{#1\right\}}

\renewcommand{\P}{\mbox{\sf P}\xspace}
\newcommand{\NP}{\mbox{\sf NP}\xspace}
\newcommand{\nP}{\mbox{\sf \#P}\xspace}

\newcommand{\CSP}{\mbox{\sf CSP}\xspace}
\newcommand{\nCSP}{\mbox{\sf \#CSP}\xspace}
\newcommand{\Z}{\mbox{\sf \#csp}\xspace}

\newcommand{\FP}{\mbox{\sf FP}\xspace}
\newcommand{\FPnP}{\ensuremath{\FP^{\mbox{\scriptsize\sf \#P}}}\xspace}
\newcommand{\nBIS}{\mbox{\sf \#BIS}\xspace}
\newcommand{\nIS}{\mbox{\sf \#IS}\xspace}
\newcommand{\BISclass}{\mbox{\#RH$\Pi_1$}\xspace}
\def\downsets{\textsc{\#Downsets}\xspace}
\def\IS{\textsc{\#IS}\xspace}
\let\epsilon=\varepsilon

\newcommand{\ORzero}{\mathop{\mathsf{OR}}}
\newcommand{\ORone}{\mathop{\mathsf{Implies}}}
\newcommand{\ORtwo}{\mathop{\mathsf{NAND}}}
\newcommand{\OR}{\ORzero}
\newcommand{\Implies}{\ORone}
\newcommand{\NAND}{\ORtwo}

\newcommand{\nSAT}{\mbox{\sf \#SAT}\xspace}

\newcommand{\IMtwo}{\mbox{$I\!M_2$}\xspace}
\newcommand{\XOR}{\mathop{\mathsf{XOR}}}

\def\redAP{\leq_\mathrm{AP}\xspace}
\def\equalAP{=_\mathrm{AP}\xspace}

\def\alphabet{\Sigma\xspace}

\def\Qhat{\widehat Q}

\def\xor{\oplus\xspace}

\makeatletter
\def\prob#1#2#3{\goodbreak\begin{list}{}{\labelwidth\z@ \itemindent-\leftmargin
                        \itemsep\z@  \topsep6\p@\@plus6\p@
                        \let\makelabel\descriptionlabel}
                \item[\it Name.]#1
                \item[\it Instance.]#2
                \item[\it Output.]#3
                \end{list}}
\makeatother

\def\poly{\mathop{\mathrm{poly}}}

\newtheorem{theorem}{Theorem}
\newtheorem{lemma}[theorem]{Lemma}
\newtheorem{corollary}[theorem]{Corollary}
\newtheorem{observation}[theorem]{Observation}

\title{An approximation
trichotomy for Boolean \nCSP
\thanks{This work was partially supported by the EPSRC
grant ``The Complexity of Counting in Constraint
Satisfaction Problems''}}

\author{Martin Dyer \\
School of Computing\\ University of Leeds\\
Leeds LS2~9JT, UK\and
Leslie Ann Goldberg \\
Department of Computer Science,\\
 University of Liverpool,\\
 Liverpool
L69 3BX, UK\and
Mark Jerrum \\
School of Mathematical Sciences,\\
Queen Mary, University of London\\
Mile End Road, London E1 4NS, UK}
\date{\today}

\begin{document}

\maketitle

\begin{abstract}
We give a trichotomy theorem for the
complexity of
approximately counting
the number of satisfying assignments of
a Boolean CSP instance.
Such problems are parameterised by a constraint language
specifying the relations that may be used in constraints.
If every relation in the constraint language
is affine then the number of
satisfying assignments can be exactly counted in polynomial
time. Otherwise, if every relation
in the constraint language
is in the co-clone \IMtwo from Post's
lattice, then the problem of
counting satisfying assignments
is complete with respect to approximation-preserving
reductions for the complexity class \BISclass. This means that
the problem of
approximately counting satisfying assignments of such a CSP instance is
equivalent in complexity to
several other known counting problems, including
the problem of approximately counting the number of independent
sets in a bipartite graph.
For every other fixed constraint language,
the problem is complete for \nP with respect to approximation-preserving
reductions, meaning that there is no \emph{fully polynomial randomised approximation
scheme} for counting satisfying assignments unless NP=RP.

\end{abstract}

\section{Introduction}

This paper gives a
trichotomy theorem for the complexity of
approximately counting
the number of satisfying assignments of
a Boolean CSP instance.
Such problems are parameterised by
a constraint language $\Gamma$
which specifies relations
that may be used in constraints.
In the Boolean case, the relations
are on a domain which has two elements.
Then $\nCSP(\Gamma)$ will denote the problem
of determining the number of (distinct) satisfying
assignments of a CSP instance with constraint
language~$\Gamma$.
Further details are given in Section 1.1 below.

Creignou and Hermann~\cite{ch} have given a dichotomy
theorem for the \emph{exact} counting problem.
They have shown that if every relation in $\Gamma$
is affine, then $\nCSP(\Gamma)$ is in \FP.
Otherwise, it is \nP-complete.
The complexity classes \FP and \nP are the
analogues of \P and \NP for counting problems.
\FP is  the class of functions computable in deterministic
polynomial time. \nP is the class of integer functions that can be
expressed as the number of accepting computations of
a polynomial-time non-deterministic Turing machine.

In this paper we build on previous work
on the complexity of approximate counting
to identify a trichotomy
in the complexity of {approximate counting}
for Boolean \nCSP.

Together with Greenhill~\cite{bis}, we have previously studied
approximation-preserving reductions
(AP-reductions) between
counting problems.
We will give details of AP-reductions in Section~\ref{sec:ap}.
For now it suffices to note that if an AP-reduction
exists from a counting problem~$f$ to a counting problem~$g$
and $g$ has a \emph{Fully Polynomial Randomised Approximation
Scheme} (FPRAS) then $f$ also has an FPRAS.

If an AP-reduction from $f$ to~$g$
exists we write $f\redAP g$, and say that
{\it $f$ is AP-reducible to~$g$}.
If $f\redAP g$ and $g\redAP f$ then we say that
{\it $f$ and $g$ are AP-interreducible}, and write $f\equalAP g$.

We previously identified~\cite{bis} three
natural classes of counting problems that are
interreducible under
AP-reductions. These are (i) those
problems that have an FPRAS,
(ii) those problems that are complete for \nP with
respect to AP-reducibility, and
a third class of intermediate complexity.
Two counting problems played a special role in~\cite{bis}.
\prob{\nSAT.}{A Boolean formula $\varphi$
in conjunctive normal form.}{The number of satisfying
assignments of~$\varphi$.}

\prob{\nBIS.}{A bipartite graph $B$.}{The number of independent sets
in~$B$.}

All problems in~$\nP$ are AP-reducible to \nSAT\
(see \cite[Section 3]{bis}).
Thus \nSAT{} is complete for $\nP$ with respect to
AP-reducibility.
This means that \nSAT{} cannot have an FPRAS unless
$\mathrm{NP}=\mathrm{RP}$.
The same is   true of any problem in $\nP$ to which \nSAT{}
is AP-reducible.

We showed in \cite[Sections 4,\,5]{bis}
that \nBIS\ is AP-interreducible
with many other natural counting problems such as
counting downsets in a partial order.
Moreover, \nBIS\ is complete
for  \BISclass, a logically-defined
subclass of \nP, with respect to AP-reductions.

The main theorem of our current paper
(Theorem~\ref{thm:main}) shows that every problem
$\nCSP(\Gamma)$ falls neatly into one of the
three classes from~\cite{bis}:
If every relation in $\Gamma$
is affine, then
trivially
$\nCSP(\Gamma)$ has an FPRAS since it is in \FP.
Otherwise, if every relation in $\Gamma$
is in a certain set \IMtwo,
then $\nCSP(\Gamma)\equalAP\nBIS$.
Otherwise $\nCSP(\Gamma)\equalAP\nSAT$.
A formal definition of
\IMtwo\
appears in Section~\ref{sec:IMtwo} --- it is the set of relations
which can be expressed as conjunctions
involving only binary implication and unary relations.

It is worth pointing out that, while every problem
$\nCSP(\Gamma)$ falls into one of the three
approximation classes from~\cite{bis},
the three classes may well not provide a partition of all
approximate counting problems in~\nP. For example, the problem of approximately counting
$3$-colourings of a bipartite graph is a problem
that may well lie between \nBIS\ and \nSAT\ in
approximability (see~\cite{bis}).

\subsection{Constraint satisfaction}

\emph{Constraint Satisfaction}, which originated in Artificial
Intelligence, provides a general framework for modelling decision
problems, and has many practical applications. (See, for
example~\cite{RoBeWa06}.) Decisions are modelled by
\emph{variables}, which are subject to \emph{constraints}, modelling
logical and resource restrictions. The paradigm is sufficiently
broad that many interesting problems can be modelled, from
satisfiability problems to scheduling problems and graph-theory
problems. Understanding the complexity of constraint satisfaction
problems has become a major and active area within computational
complexity~\cite{cks, hnbook}.

A Constraint Satisfaction Problem (CSP)
typically has a finite \emph{domain}, which we
denote by $\set{0,\ldots,q-1}$ for a positive integer~$q$.
In this paper we are
interested in the \emph{Boolean} case
$q=2$.
A \emph{constraint
language\/} $\Gamma$
with domain $\{0,\ldots,q-1\}$
is a set of relations on~$\{0,\dots,q-1\}$. For example,
take
$q=2$.
The
relation $R=\{(0,0,1)$, $(0,1,0)$, $(1,0,0)$, $(1,1,1)\}$ is a 3-ary
relation on the domain $\{0,1\}$, with four tuples.

Once we have fixed a constraint language $\Gamma$,
an \emph{instance\/} of the CSP is a set of \emph{variables\/}
$V=\{v_1,\ldots,v_n\}$ and a set of \emph{constraints}. Each
constraint has a \emph{scope,} which is a tuple of variables (for
example, $(v_4, v_5, v_1)$) and a relation from~$\Gamma$ of the same
arity, which constrains the variables in the scope. An
\emph{assignment} $\sigma$ is a function from~$V$
to~$\{0,\ldots,q-1\}$. The
assignment~$\sigma$ is \emph{satisfying} if the scope of every
constraint is mapped to a tuple that is in the corresponding
relation.
In our example above, an assignment~$\sigma$ satisfies the constraint with scope
$(v_4,v_5,v_1)$ and relation~$R$, written $R(v_{4},v_{5},v_{1})$,
if and only if it maps an odd number of the variables in
$\{v_1,v_4,v_5\}$ to the value~$1$.
Given an instance~$I$ of a CSP with
constraint language $\Gamma$, the \emph{decision problem}
$\CSP(\Gamma$) asks us to determine whether any assignment
satisfies~$I$. The \emph{counting problem} $\nCSP(\Gamma$) asks us to
determine the \emph{number} of (distinct) satisfying
assignments of~$I$, which we will denote by $\Z(I)$.

Varying the constraint language~$\Gamma$ defines the classes \CSP and
\nCSP of decision and counting problems. These contain problems of
different computational complexities.
For example,
consider the binary relations defined by
$\ORzero=\{(0,1),(1,0),(1,1)\}$,
$\ORone=\{(0,0),(0,1),(1,1)\}$, and
$\ORtwo=\{(0,0),(0,1),(1,0)\}$.
If
$\Gamma=\{\ORzero,\ORone,\ORtwo\}$
then
$\CSP(\Gamma)$ is the classical 2-Satisfiability problem, which is in
\P. On the other hand, there is a similar constraint
language~$\Gamma'$ with four relations of arity~3 such that
3-Satisfiability (which is \NP-complete) can be represented in
$\CSP(\Gamma')$.  It may happen, as here, that the counting problem
is harder than the decision problem:  $\nCSP(\Gamma)$
contains the problem of counting independent sets in graph,
and is thus \nP-complete.

Any  decision problem $\CSP(\Gamma)$ is in \NP, but not every problem
in \NP can be represented as a CSP. For example, the question ``Is
$G$ Hamiltonian?'' cannot be expressed as a CSP, because the
property of being Hamiltonian cannot be captured by relations of
bounded size. This limitation of the class \CSP has an important
advantage. If $\P \neq \NP$, then there are problems
which are neither in \P nor \NP-complete~\cite{L75}. But, for
well-behaved smaller classes of decision problems, the situation can
be simpler. We may have a \emph{dichotomy theorem}, partitioning all
problems in the class into those which are in \P and those which are
\NP-complete. There are no ``leftover'' problems of intermediate
complexity. It has been conjectured that there is a dichotomy
theorem for \CSP.  The conjecture is  that $\CSP(\Gamma)$ is in \P
for some constraint languages $\Gamma$, and $\CSP(\Gamma)$ is \NP-complete
for all other constraint languages $\Gamma$. This conjecture appeared in a
seminal paper of Feder and Vardi~\cite{fv}, but has not yet been proved.
A similar
dichotomy, between \FP and \nP-complete, is conjectured for
\#CSP~\cite{BD}.
Recently, Bulatov~\cite{Bul} has announced a positive
resolution of this conjecture.

There have been many
important results for subclasses of \CSP and \nCSP. We mention the
most relevant to our paper here. The first decision dichotomy was that
of Schaefer~\cite{schaefer}, for the Boolean domain $\{0,1\}$.
Schaefer's result is as follows.
\begin{theorem}[Schaefer~\cite{schaefer}]
\label{thm:schaefer} Let $\Gamma$ be a constraint language with domain
$\{0,1\}$. The problem $\CSP(\Gamma)$ is
in \P if\/ $\Gamma$ satisfies one of the conditions below. Otherwise,
$\CSP(\Gamma)$ is \NP-complete.
\begin{enumerate}[topsep=5pt]
\item $\Gamma$ is $0$-valid or $1$-valid.
\item $\Gamma$ is weakly positive or weakly negative.
\item $\Gamma$ is affine.
\item $\Gamma$ is bijunctive.
\end{enumerate}
\end{theorem}

We will not give detailed definitions of the conditions in
Theorem~\ref{thm:schaefer}, but the interested reader is referred to
the paper~\cite{schaefer} or to Theorem~6.2 of the
textbook~\cite{cks}. An interesting feature is that the conditions
in~\cite[Theorem~6.2]{cks} are all checkable. That is, there is an algorithm to
determine whether $\CSP(\Gamma$) is in \P or \NP-complete, given a constraint
language~$\Gamma$ with domain~$\{0,1\}$.
We say in this case that the dichotomy is \emph{effective}.

A Boolean relation $R$ is
said to be \emph{affine} if the set of
tuples~$x\in R$ is the set of
solutions to a system of linear
equations over GF($2$).
Creignou and Hermann~\cite{ch} adapted Schaefer's decision
dichotomy to obtain a counting dichotomy for the Boolean domain.
Their result is as follows.
\begin{theorem}[Creignou and Hermann~\cite{ch}]
\label{thm:CH}
Let $\Gamma$ be a constraint language with domain
$\{0,1\}$.
The problem $\nCSP(\Gamma)$ is in\/ \FP\/
if\/ every
relation in $\Gamma$ is affine.
Otherwise, $\nCSP(\Gamma)$ is \nP-complete.
\end{theorem}

Creignou and Hermann's result is an important starting point
for our work, and we will discuss it further
below. Note that
there is an algorithm for determining
whether a relation is affine, so the dichotomy is effective.

We have recently~\cite{bcsp} extended
Creignou and Hermann's dichotomy
to the domain of \emph{weighted} Boolean $\nCSP$
giving an effective dichotomy between \FP and \FPnP\
for the problem of computing the partition function
of a weighted Boolean CSP instance.

\subsection{The complexity of approximate counting}
\label{sec:ap}

We now recall the necessary background from~\cite{bis}.
A \emph{randomised approximation scheme\/} is an algorithm for
approximately computing the value of a
function~$f:\Sigma^*\rightarrow
\mathbb{N}$. The
approximation scheme has a parameter~$\varepsilon>0$
which specifies
the error tolerance.
A \emph{randomised approximation scheme\/} for~$f$ is a
randomised algorithm that takes as input an instance $ x\in
\alphabet^{\ast }$ (e.g., an encoding of a CSP
instance) and an error
tolerance $\varepsilon >0$, and outputs an
integer $z$
(a random variable on the ``coin tosses'' made by the algorithm)
such that, for every instance~$x$,
\begin{equation}
\label{eq:3:FPRASerrorprob}
\Pr \big[e^{-\epsilon} f(x)\leq z \leq
e^\epsilon f(x)\big]\geq \frac{3}{4}\, .
\end{equation}
The randomised approximation scheme is said to be a
\emph{fully polynomial randomised approximation scheme},
or \emph{FPRAS},
if it runs in time bounded by a polynomial
in $ |x| $ and $ \epsilon^{-1} $.
(See Mitzenmacher and Upfal~\cite[Definition 10.2]{MU05}.)
Note that the quantity $3/4$ in
Equation~(\ref{eq:3:FPRASerrorprob})
could be changed to any value in the open
interval $(\frac12,1)$ without changing the set of problems
that have randomised
approximation schemes~\cite[Lemma~6.1]{JVV86}.

Suppose that $f$ and $g$ are functions from
$\alphabet^{\ast }$ to~$\mathbb{N}$. An
``approximation-preserving
reduction''
(AP-reduction)
from~$f$ to~$g$ gives a way to turn an FPRAS for~$g$
into an FPRAS for~$f$. An AP-reduction
from $f$ to~$g$ is a randomised algorithm~$\mathcal{A}$ for
computing~$f$ using an oracle for~$g$\footnote{The
reader who is not familiar with oracle Turing machines
can just think of this as an imaginary (unwritten)
subroutine for computing~$g$.}.
The algorithm~$\mathcal{A}$ takes
as input a pair $(x,\varepsilon)\in\alphabet^*\times(0,1)$, and
satisfies the following three conditions: (i)~every oracle call made
by~$\mathcal{A}$ is of the form $(w,\delta)$, where
$w\in\alphabet^*$ is an instance of~$g$, and $0<\delta<1$ is an
error bound satisfying $\delta^{-1}\leq\poly(|x|,
\varepsilon^{-1})$; (ii) the algorithm~$\mathcal{A}$ meets the
specification for being a randomised approximation scheme for~$f$
(as described above) whenever the oracle meets the specification for
being a randomised approximation scheme for~$g$; and (iii)~the
run-time of~$\mathcal{A}$ is polynomial in $|x|$ and
$\varepsilon^{-1}$.  In formulating a definition of 
approximation-presering reduction, a number of choices
must be faced.  The key requirement is that 
the class of functions computable by an FPRAS should be closed 
under AP-reducibility.
Informally, we have gone for the most liberal notion
of reduction meeting this requirement.

\subsection{Notation for relations}

Define the unary relations
$\delta_0=\{(0)\}$ and
$\delta_1=\{(1)\}$.
Recall the binary relation
$\ORone=\{(0,0),(0,1),(1,1)\}$.

For convenience, according to context,
we view a $k$-ary relation $R$ either as a set of $k$-tuples
or as a $k$-ary predicate.  Thus the notations 
$R(x_1,\ldots,x_k)=1$ (or just $R(x_1,\ldots,x_k)$) and $(x_1,\ldots,x_k)\in R$
are equivalent.
For example,
${\delta_0}(x) = \overline{x}$,
${\delta_1}(x)=x$ and ${\ORone}(x,y) = \overline{x}\vee y$.

\subsection{The set of relations \IMtwo}
\label{sec:IMtwo}

An $n$-ary relation $R$ is in \IMtwo\
if and only if $R(x_1,\ldots,x_n)$ is logically equivalent
to a conjunction of predicates of the form
$\delta_0(x_i)$, $\delta_1(x_i)$
and $\Implies(x_i,x_j)$.

As we will discuss below,
Creignou, Kolaitis, and Zanuttini
\cite{ckz} have shown that \IMtwo\
is a co-clone
in Post's lattice (see \cite{brsv}).

\subsection{Our result}

We can now state our main theorem.

\begin{theorem}
Let $\Gamma$ be a   constraint language with
domain $\{0,1\}$.
If every relation in $\Gamma$ is affine
then $\nCSP(\Gamma)$ is in \FP.
Otherwise if every relation in $\Gamma$ is
in \IMtwo\ then $\nCSP(\Gamma)\equalAP\nBIS$.
Otherwise $\nCSP(\Gamma)\equalAP\nSAT$.
\label{thm:main}
\end{theorem}

The main ingredients in the proof are:
(1) the AP-reduction technology of~\cite{bis}, which allows us to
effectively ``pin'' certain CSP variables in hardness proofs
(see Section~\ref{sec:pinning}); (2) the ``implementations''
of Creignou, Khanna and Sudan~\cite{cks}, which show how to
construct the key relations $\OR$, $\Implies$, and $\NAND$ from a non-affine relation
and $\delta_0$ or $\delta_1$
(see
Section~\ref{sec:implementation});
(3) the complexity class \BISclass\ from~\cite{bis}, consisting of
those problems which are AP-interreducible with \nBIS; and
(4) the co-clone \IMtwo\ in Post's lattice (see Section~\ref{sec:PolyIMtwo}),
since the complexity of $\nCSP(\Gamma)$ for $\Gamma \subseteq \IMtwo$
turns out to be closely
connected to the complexity of \nBIS.

\section{The pieces of the proof}

\subsection{Types of relations}

A relation $R$ is \emph{0-valid} if
the all-zero tuple is in~$R$.
Similarly, $R$ is \emph{1-valid} if
the all-ones tuple is in~$R$.
Following~\cite{cks}, we say that
a $k$-ary relation $R$ is
\emph{complement-closed} (C-closed in~\cite{cks})
if
$$(x_1,\dots,x_k)\in R \Leftrightarrow
(x_1 \oplus 1,\ldots,x_k\oplus 1)\in R,$$
where $\oplus$ is the exclusive or operator.

We say that $\Gamma$ is 0-valid if every $R\in \Gamma$
is 0-valid and we define what it means for $\Gamma$ to
be 1-valid or complement-closed similarly.

\subsection{Some preliminary complexity
results}

We start by observing that every problem
$\nCSP(\Gamma)$ is AP-reducible to \nSAT
\begin{observation}
\label{obs:SATeasy}
Let $\Gamma$ be a  constraint language with
domain $\{0,1\}$.
Then $\nCSP(\Gamma)\redAP\nSAT$.
\end{observation}
Observation~\ref{obs:SATeasy} follows
from the fact that all problems in \nP\
are AP-reducible to \nSAT \cite{bis}.
Another, very simple, but useful, observation is
the following.
\begin{observation}
\label{obs:ignore}
Let $\Gamma$ be a constraint language with domain
$\{0,1\}$.
Suppose $\Gamma'\subseteq \Gamma$.
Then $\nCSP(\Gamma')\redAP\nCSP(\Gamma)$.
\end{observation}
Observation~\ref{obs:ignore} is true for
the simple reason that every instance of
$\nCSP(\Gamma')$ \emph{is} an instance
of $\nCSP(\Gamma)$.

Recall the relations
$\ORzero=\{(0,1),(1,0),(1,1)\}$ and
$\ORtwo=\{(0,0),(0,1),\allowbreak(1,0)\}$.
These relations
are particularly fundamental for us, and we start with
complexity results about these.

\begin{lemma}
\label{lem:SATtwo}
$\nSAT\redAP \nCSP(\{\ORtwo\})$.
\end{lemma}

\begin{proof}

It was shown in~\cite{bis}
that the following problem is AP-interreducible with \nSAT.

\prob{\IS.}{A graph $G$.}
{The number of independent sets
in~$G$.}

We show that $\nIS \redAP \nCSP(\{\ORtwo\})$.
Let $G=(V,E)$ be an instance of $\nIS$.
Construct an instance $I$ of $\nCSP(\{\ORtwo\})$
with variable set~$V$.
For every edge $(u,v)\in E$, add constraint $\ORtwo(u,v)$.
There is now a bijection between independent sets of~$G$
and satisfying assignments~$\sigma$ of~$I$:
variables $v$ with $\sigma(v)=1$
correspond to vertices in the independent set.

\end{proof}

\begin{lemma}
\label{lem:SATzero}
$\nSAT \redAP\nCSP(\{\ORzero\}) $.
\end{lemma}

\begin{proof}

The proof that $\nIS \redAP \nCSP(\{\ORzero\})$ is similar
(just associate variables $v$ with $\sigma(v)=1$ with
vertices that are \emph{out} of
the independent set).

\end{proof}

Finally, we will need a couple of complexity results
involving $\nBIS$.

\begin{lemma}
\label{lem:BISone}
$\nBIS \redAP \nCSP(\{\ORone\})$.
\end{lemma}

\begin{proof}

Let $G$ be an instance of $\nBIS$ with vertex sets~$U$ and~$V$ and edge set~$E$.
Construct an instance $I$ of $\nCSP(\{\ORone\})$ with variable set~$U \cup V$.
For every edge $(u,v)\in E$ with $u\in U$  add constraint $\ORone(u,v)$.
There is now a bijection between independent sets of~$G$
and satisfying assignments~$\sigma$ of~$I$: a variable $u\in U$ with $\sigma(u)=1$ is in the independent
set and a variable $v\in V$ with $\sigma(v)=0$ is in the independent set.

\end{proof}

\begin{lemma}
\label{lem:downsets}
Suppose $\Gamma \subseteq \IMtwo$. Then
$\nCSP(\Gamma)\redAP\nBIS$.
\end{lemma}

\begin{proof}

It is straightforward to show that $\nCSP(\Gamma)$
is in the complexity class \BISclass\ which
has \nBIS\ as a complete problem \cite{bis}.

However, to avoid giving a definition of \BISclass,
which requires some notation, we will instead
show $\nCSP(\Gamma)\redAP \downsets$,
where \downsets\ is the following
counting problem which
was shown in~\cite{bis}
to be AP-interreducible with \nBIS.

\prob{\downsets.}{A partially ordered set
$(X,{\preceq})$.}{The number
of downsets\footnote{A {\it downset\/} in $(X,{\preceq})$ is a subset
$D\subseteq X$ that is closed under $\preceq$; i.e., $x\preceq y$ and
$y\in D$ implies $x\in D$.}  in $(X,{\preceq})$.}

Consider an instance~$I$ of $\nCSP(\Gamma)$
with variables $v_1,\ldots,v_n$.
The set of constraints can be viewed as an equivalent set of
constraints of the form
$\delta_0(v_i)$, $\delta_1(v_i)$
or $\Implies(v_i,v_j)$.
Denote by $\Implies^*$ the transitive closure
of the $\Implies$ relation on $\{v_1,\ldots,v_n\}$:
thus $\Implies^*(v_i,v_j)$ if there is a sequence of
variables, starting with $v_i$ and ending with $v_j$,
such that every adjacent pair in the sequence is
constrained by $\Implies$.

Let
$N_0(I)$ be the set of variables~$v_i$ for which either
(i)~a constraint $\delta_0(v_i)$ occurs in~$I$,
or  (ii)~there exists a variable $v_j$ such that
$\Implies^*(v_i,v_j)$ and a constraint $\delta_0(v_j)$
occurs in~$I$.
These are the variables that are forced to be~$0$
in any satisfying assignment of~$I$.
Define $N_1(I)$ analogously to be the set of variables
that are forced to be~$1$ in any satisfying assignment.
We can assume without loss of generality that $N_0(I)$
and $N_1(I)$
are disjoint. Otherwise the instance~$I$ has no satisfying
assignments, and we can determine this without even using
the downsets oracle.

Now remove all the variables in $N_0(I)$ and $N_1(I)$ from the instance~$I$:
this does not affect the number of satisfying assignments, since
these variables do not constrain any of the others.
Also identify all pairs of variables $v_i,v_j$ such that
$\Implies^*(v_i,v_j)$ and $\Implies^*(v_j,v_i)$:  again, this
does not affect the number of satisfying assignments.

The remaining variables and relations define a partial order
$(X,{\preceq})$ since our construction forces antisymmetry.
The satisfying assignments of~$I$ correspond $1$--$1$ with
the downsets of $(X,{\preceq})$.
\end{proof}

\subsection{A useful tool: pinning}
\label{sec:pinning}

\emph{Pinning}
is the ability to tie certain CSP variables
to specific values in hardness proofs.
This idea was used by Creignou and Hermann in their
dichotomy theorem~\cite{ch}.
Similar ideas have been used in
many other hardness proofs and dichotomy
theorems~\cite{BD, BG05, bcsp, DG}.
As we show in this section, AP-reductions
facilitate a particularly useful form of pinning.

\begin{lemma}
\label{powerpin}
Let $\Gamma$ be a constraint language with domain
$\{0,1\}$.
Suppose there is a relation $R\in \Gamma$ for
which, for some position $j$, $R$ has more tuples $t$
with $t_j=0$ than with $t_j=1$. Then
$\nCSP(\Gamma\cup\{\delta_0\}) \redAP
\nCSP(\Gamma)$.
Similarly, if there is a relation $R\in \Gamma$
for which, for some position $j$, $R$ has
more tuples $t$ with $t_j=1$ than with $t_j=0$ then
$\nCSP(\Gamma\cup\{\delta_1\}) \redAP
\nCSP(\Gamma)$.
\end{lemma}

\begin{proof}
Consider an instance $I$ of
$\nCSP(\Gamma\cup\{\delta_0\})$
with $n$ variables.
Suppose there is an arity-$k$ relation $R\in \Gamma$ for
which, for position $j$, $R$ has $w$ tuples~$t$
with $t_j=0$ and $w'<w$ tuples~$t$ with $t_j=1$.

As in the proof of Lemma~\ref{lem:downsets}, let
$N_0(I)$ be the set of variables~$x$ to which
one or more constraints $\delta_0(x)$ occurs in~$I$
and let $N_1(I)$ be the set of variables~$y$ to
which one or more constraints $\delta_1(y)$ occurs.
Let $n_0=|N_0(I)|$.
Let $m=\lceil
(n+2)/\lg(w/w')
\rceil$.
Construct an instance $I'$ of
$\nCSP(\Gamma)$. Include all constraints in~$I$
other than those involving $\delta_0$.
For each variable $x\in N_0(I)$, and every $a \in \{1,\ldots,m\}$,
introduce $k-1$ new variables $x'_{a,b}$ for $b\in\{1,\ldots,k\}-\{j\}$.
Introduce a new constraint in $I'$ with relation $R$ and
variable $x$ in the $j\,$th position, and
$x'_{a,b}$ in the $b\,$th position, for all~$b$.

Now a satisfying assignment for $I$
can be extended in $w^{m n_0}$ ways to satisfying
assignments of~$I'$.
An assignment for $I$ that violates
one of the $\delta_0(x)$ constraints can be extended in at most
$w^{m (n_0-1)} {w'\,}^m$ ways to
satisfying assignments
of~$I'$.
Thus,
$$\Z(I) w^{m n_0}
 \leq \Z(I') \leq
\Z(I) w^{m n_0} + 2^n w^{m (n_0-1)} {w'\,}^m,$$
i.e.,
$$
\Z(I) 
 \leq \frac{\Z(I')}{w^{mn_{0}}} \leq
\Z(I) + 2^n(w'/w)^m.
$$
So, by definition of~$m$,
$$\Z(I) \leq \frac{\Z(I')}{w^{m n_0}}\leq \Z(I)+\frac{1}{4}.$$

Thus we have constructed a reduction
from $\nCSP(\Gamma\cup\{\delta_0\})$
to $\nCSP(\Gamma)$:
Given an instance~$I$ of~$\nCSP(\Gamma\cup\{\delta_0\})$,
use an oracle for $\nCSP(\Gamma)$
to approximate $\Z(I')$,
divide by $w^{m n_0}$, and
round
to the nearest integer (always down).
Note that the reduction makes only one oracle call
(and uses no randomisation).

To show that the reduction is indeed an AP-reduction,
we add some technical details concerning the choice of
the accuracy parameter~$\delta$ in the
oracle call (see the definition of AP-reduction
in Section~\ref{sec:ap}). These details are here to make
the proof   complete, but they are not essential
for understanding the rest of the paper.

If we had
$$\Z(I) = \frac{\Z(I')}{w^{m n_0}},$$
we could
simply set $\delta=\epsilon$, since division by a constant
preserves relative error.
Instead we have
$$\Z(I) = \left\lfloor\frac{\Z(I')}{w^{m n_0}}\right\rfloor.$$
The discontinuous floor function could
spoil the approximation when its argument is small.

The situation here is that the true answer
$N=\Z(I)$ is obtained by rounding the
fraction $Q=\frac{\Z(I')}{w^{m n_0}}$
where we have $|Q-N|\leq1/4$.

Suppose that the oracle
provides an approximation~$\Qhat$ to~$Q$ satisfying
$Qe^{-\delta}\leq \Qhat\leq Qe^{\delta}$
(as it is required to do
with probability at least $3/4$).
Set $\delta=\epsilon/21$,
where $\epsilon$ is the accuracy parameter
governing the final result.
There are two cases.  If $N\leq 2/\epsilon$, then a short
calculation yields $|\Qhat-Q|<1/4$ implying that the result returned
by the algorithm is exact.  
If $N> 2/\epsilon$, then the result returned
is in the range $[(N-1/4)e^{-\delta}-1/2,(N+1/4)e^{\delta}+1/2]$
which, for the chosen $\delta$, is contained in
$[Ne^{-\epsilon},Ne^{\epsilon}]$.

Thus, we have an AP-reduction from
$\nCSP(\Gamma\cup\{\delta_0\})$ to
$\nCSP(\Gamma)$.
The reduction showing
$\nCSP(\Gamma\cup\{\delta_1\}) \redAP
\nCSP(\Gamma)$ is similar.
\end{proof}

\subsection{Affine relations}

We use the following well-known facts about affine relations.
\begin{lemma}

\begin{enumerate}
\item A $k$-ary Boolean relation~$R$ is affine if and only if
$a,b,c\in R$ implies $d=a\oplus b \oplus c\in R$,
where the $\oplus$ operator is applied componentwise.
\item If $R$ is not affine,
then for any fixed $a\in R$
there are $b,c\in R$ such that $a \oplus b \oplus c \not\in R$.
\item If $R$ is not affine, then there are $a,b$ in $R$
such that $a \oplus b\not\in R$.
\end{enumerate}
\label{affine:facts}
\end{lemma}
\begin{proof}
For Part~(i) see, for example, Lemma~4.10 of~\cite{cks}).
Part~(ii) is proved in the same place, but since it is
a little less well-known, we provide the proof:
Suppose the contrary that $R$ is not affine, but for all $b,c\in R$,
$a \oplus b \oplus c\in R$.
Choose $s_0,s_1,s_2 \in R$ such that $s_0 \oplus s_1 \oplus s_2 \not\in R$.
From $b=s_0$, $c=s_1$, $d=a \oplus s_0 \oplus s_1$
we have $d \in R$.
From $b=s_2$, $c=d$ we have $a \oplus s_2 \oplus d = s_0 \oplus s_1 \oplus s_2 \in R$,
a contradiction.

To see Part~(iii), note that the condition
``$\forall a,b:$ $a,b\in R$ implies $a\oplus b \in R$''
implies that $R$ is affine, so, if $R$ is not affine then
the condition is false.
\end{proof}

\subsection{Implementation}
\label{sec:implementation}

Let $\Gamma$ be a constraint language with domain
$\{0,1\}$. $\Gamma$ is said to
\emph{implement}\footnote{There are many
variants of ``implement'' defined in the literature.
See \cite[Chapter~5]{cks}, where the kind of implementation
we define here is called ``faithful'' and ``perfect''.}
a $k$-ary relation~$R$
if, for some $k'\geq k$
there is a CSP instance~$I$
with variables $x_1,\ldots,x_{k'}$
and constraints in~$\Gamma$
such that, for every tuple $(s_1,\ldots,s_k)\in R$,
there is exactly one satisfying
assignment~$\sigma$ of~$I$ with
$\sigma(x_1)=s_1,\ldots,\sigma(x_k)=s_k$
and for every
tuple $(s_1,\ldots,s_k)\not\in R$,
there are no satisfying
assignments~$\sigma$ of~$I$ with
$\sigma(x_1)=s_1,\ldots,\sigma(x_k)=s_k$.
Note the following straightforward observation,
which is essentially a parsimonious reduction~\cite[p.441]{Papa94}.
\begin{observation}
\label{obs:implement}
If\/ $\Gamma$ implements~$R$
then $\nCSP(\Gamma\cup \{R\}) \redAP
\nCSP(\Gamma)$.
\end{observation}

We will use several implementations of
Creignou, Khanna and Sudan.
Proofs are provided in the appendix in order to make the paper self-contained.

\begin{lemma}
\label{implementone}
(Creignou, Khanna and Sudan, \cite[Lemmas 5.24 and 5.25]{cks})
Let $\Gamma$ be a constraint language with domain
$\{0,1\}$.
\begin{enumerate}
\item
If\/ $\Gamma$ contains a relation $R$ that
is 0-valid, 1-valid and not complement-closed then
$\Gamma$ implements the relation
$R'=\{(0,0),(1,1),(1,0)\}$.
\item If\/ $\Gamma$ contains a relation $R$ that
is not 0-valid, not 1-valid and not comple\-ment-closed then
$\Gamma$ implements~$\delta_0$ and~$\delta_1$.
\item If\/ $\Gamma$ contains a relation $R$ that
is 0-valid and not 1-valid then
$\Gamma$ implements~$\delta_0$.
\item If\/
$\Gamma$ contains a relation~$R$ that
is 1-valid and not 0-valid then
$\Gamma$ implements~$\delta_1$.
\end{enumerate}
\end{lemma}

\begin{lemma} (Creignou, Khanna and Sudan,
\cite[Claim 5.31]{cks})
\label{lem12}
\label{cases}
Let $R$ be a ternary relation
containing $(0,0,0)$, $(0,1,1)$ and $(1,0,1)$
but not $(1,1,0)$. Then
$\{R,\delta_0\}$ implements
one of $\ORone$ and $\ORtwo$.
\end{lemma}

\begin{lemma}(Creignou, Khanna and Sudan, \cite[Lemma 5.30]{cks})
\label{lemsim}
\label{lem11}
If $R$ is a relation over $\{0,1\}$ that is
not affine then
$\{R,\delta_0\}$ implements one of $\ORzero$, $\ORone$, and $\ORtwo$
and so does $\{R, \delta_1\}$.
\end{lemma}

\subsection{Pinning revisited}
\label{sec:morepinning}

Combining the useful pinning that we get from AP-reductions
(Lemma~\ref{powerpin})
with the implementations of $\OR$, $\Implies$ and $\NAND$ in Section~\ref{sec:implementation},
we obtain a useful lemma which says that we can \emph{always} do some pinning.

\begin{lemma}
\label{otherpin}
Let $\Gamma$ be a constraint language with domain
$\{0,1\}$.
Then either
$\nCSP(\Gamma\cup\{\delta_0\}) \redAP
\nCSP(\Gamma)$ or
$\nCSP(\Gamma\cup\{\delta_1\}) \redAP
\nCSP(\Gamma)$ (or both).
\end{lemma}

\begin{proof}

First, suppose that $\Gamma$ is not complement-closed.
If $\Gamma$ contains a relation $R$ that
is not 0-valid, not 1-valid and not complement-closed
then we finish by Observation~\ref{obs:implement}
and Part~(ii) of Lemma~\ref{implementone}.
If $\Gamma$ contains a relation $R$ that is
0-valid, 1-valid and not complement-closed
then it implements the relation $R'$ from Part~(i) of Lemma~\ref{implementone}
so by Observation~\ref{obs:implement},
$\nCSP(\Gamma \cup \{R'\}) \redAP \nCSP(\Gamma)$.
But Lemma~\ref{powerpin}
shows both
$\nCSP(\Gamma \cup \{R',\delta_0\}) \redAP \nCSP(\Gamma \cup \{R'\})$
and
$\nCSP(\Gamma \cup \{R',\delta_1\}) \redAP \nCSP(\Gamma \cup \{R'\})$.
Otherwise $\Gamma$ contains a relation $R$ that
is 0-valid and not 1-valid (or vice-versa) and we finish by
Part~(iii) (or Part~(iv)) of Lemma~\ref{implementone},
and Observation~\ref{obs:implement}.

Second (and finally), suppose that $\Gamma$ is complement-closed.
Here is a simple AP-reduction from
$\nCSP(\Gamma\cup \{\delta_0\})$ to~$\nCSP(\Gamma)$.
Let $I$ be an instance of $\nCSP(\Gamma\cup \{\delta_0\})$.
Construct an instance $I'$ of $\nCSP(\Gamma)$
by adding a new variable~$z_0$.
For all $x\in N_0(I)$ (all variables~$x$
to which
one or more constraints $\delta_0(x)$ in~$I$ apply),
replace all occurrences of variable~$x$ with $z_0$ in~$I'$.
Now note that $2 \Z(I)=\Z(I')$ since there
is a one-to-two map from satisfying assignments of~$I$
and satisfying assignments of~$I'$.
In particular, if $s$ is an assignment to
all variables of~$I$ other than those in $N_0(I)$
and
$s$
is satisfying, provided the rest of the
variables are assigned value~$0$, then $s$ is mapped to $s;z_0=0$ and
$\overline{s};z_0=1$, where $\overline{s}$ is the tuple obtained
from~$s$ by complementing the assignment of every variable.
Both satisfy $I'$ since $\Gamma$
is complement-closed.  It is clear that all satisfying assignments
of~$I'$ arise in this way.
\end{proof}

\subsection{Notation for Boolean functions}
The following definitions are from \cite{bcrv, brsv}.
An $m$-ary Boolean function $f$ is \emph{monotonic}
if and only if
$(a_1,\ldots,a_m) \leq (b_1,\ldots,b_m)$ componentwise
implies $f(a_1,\ldots,a_m) \leq f(b_1,\ldots,b_m)$.
Let $M_2$ be the set of all monotone Boolean
functions $f$ satisfying $f(0,\ldots,0)=0$
and $f(1,\ldots,1)=1$.
Given a set $B$ of Boolean functions, the closure $[B]$
consists of all functions
that can be defined by propositional formulas with
connectives from~$B$ (see~\cite{bcrv}).

An $m$-ary Boolean function $f$ is said to
be a \emph{polymorphism} of an $n$-ary relation $R(x_1,\ldots,x_n)$
if applying $f$ componentwise to $m$ tuples in $R$
results in a tuple that is also in~$R$.

\subsection{Polymorphisms and \IMtwo}
\label{sec:PolyIMtwo}

In the terminology of universal algebra,
Creignou, Kolaitis, and Zanuttini
\cite{ckz} have shown that \IMtwo\
is precisely the co-clone corresponding to $M_2$,
which is a clone
in Post's lattice (see \cite{brsv}).
The direction of this result
that we will use is the following.

\begin{lemma} (Creignou, Kolaitis, Zanuttini, \cite{ckz})
If the relation $R$ is not in \IMtwo then
there is an $f\in M_2$
that is not a polymorphism of $R$.
\label{lemI}
\end{lemma}

\begin{corollary} If the $n$-ary
relation $R$ is not in
\IMtwo then there are Boolean tuples
$(a_1,\ldots,a_n)\in R$ and
$(b_1,\ldots,b_n)\in R$ such that
either
$(a_1 \wedge b_1,\ldots,a_n \wedge b_n)\not\in R$
or
$(a_1 \vee b_1,\ldots,a_n \vee b_n)\not\in R$
(or both).
\label{corI}
\end{corollary}

\begin{proof}
We will use the fact
(see \cite{bcrv})
that $M_2=[\{\vee,\wedge\}]$
where $x\vee y$ is the OR of the Boolean values~$x$
and~$y$ and $x\wedge y$ is the AND of $x$ and~$y$.
Thus, every function $f\in M_2$ can be defined
by a propositional formula using the 2-ary connectives
$\vee$ and $\wedge$.

The proof is by induction on the
number of connectives used in the propositional formula
used to represent the function $f$
from Lemma~\ref{lemI}.

The case $f(x)=x$ (in which $f$ has no connectives)
cannot arise since
the identity function is a polymorphism of every relation.
The cases $f(x,y)=x\vee y$ and $f(x,y)=x\wedge y$
(in which $f$ has one connective)
immediately give the corollary.

For the inductive step, we assume either
$f(x_1,\ldots,x_m)=f'(x_1,\ldots,x_m)\vee f''(x_1,\ldots,x_m)$
or
$f(x_1,\ldots,x_m)=f'(x_1,\ldots,x_m)\wedge f''(x_1,\ldots,x_m)$
where $f'$ and $f''$ have fewer connectives than~$f$.
Note that $f'$ and $f''$ may not actually use all of the
variables in $x_1,\ldots,x_m$.

These two cases are similar, so
suppose we are in the first of them.
That is, suppose
$$f(x_1,\ldots,x_m)=f'(x_1,\ldots,x_m)\vee f''(x_1,\ldots,x_m).$$
Suppose also that $f'$ and $f''$ are polymorphisms of~$R$
(otherwise we will apply the inductive hypothesis
to one of these functions which has fewer connectives).
Let $t^1,\ldots,t^m$ be $m$ $n$-tuples in~$R$,
such that the tuple obtained by applying 
$f$ componentwise to $t^{1},\ldots,t^{m}$ is not in~$R$.
Let $t'$ be the $n$-tuple obtained by applying
$f'$ componentwise to $t^1,\ldots,t^m$
and let $t''$ be the $n$-tuple obtained by applying
$f''$ componentwise to $t^1,\ldots,t^m$.
Since $f'$ and $f''$ are polymorphisms of~$R$,
we know that $t'$ and $t''$ are in~$R$.
However, since $f$ is not a polymorphism of~$R$, the
tuple
$t'\vee t''$ is not in~$R$, proving
the corollary.

\end{proof}

\section{Putting it all together: the proof of
Theorem~\ref{thm:main}}

We start with a lemma establishing a reduction
from \nSAT.

\begin{lemma}\label{lem:main}
Let $R_1$ and $R_2$ be relations on $\{0,1\}$.
If $R_1$ is not affine and $R_2$ is not in \IMtwo
then $\nSAT\redAP\nCSP(\{R_1,R_2\})$.
\end{lemma}

\begin{proof}

Apply Lemma~\ref{otherpin} with $\Gamma=\{R_1,R_2\}$.
Then either
$\nCSP(\{R_1,R_2,\delta_0\})\redAP
\nCSP(\{R_1,R_2\})$ or
$\nCSP(\{R_1,R_2,\delta_1\})\redAP
\nCSP(\{R_1,R_2\})$.
Assume the former (the latter case is symmetric).

Now use Lemma~\ref{lem11} together with
Observation~\ref{obs:implement}.
Since $R_1$ is not affine this shows one of the following.
\begin{itemize}
\item
$\nCSP(\{R_1,R_2,\delta_0,\ORzero\})\redAP
\nCSP(\{R_1,R_2,\delta_0\})$, or
\item
$\nCSP(\{R_1,R_2,\delta_0,\ORtwo\})\redAP
\nCSP(\{R_1,R_2,\delta_0\})$, or
\item
$\nCSP(\{R_1,R_2,\delta_0,\ORone\})\redAP
\nCSP(\{R_1,R_2,\delta_0\})$.
\end{itemize}

In the first two of these cases, we are finished
by Observation~\ref{obs:ignore} and
Lemmas~\ref{lem:SATtwo} and~\ref{lem:SATzero},
so assume
the final case.
Using Lemma~\ref{powerpin}
with the second position of~$\ORone$,
we get
$\nCSP(\{R_1,R_2,\delta_0,\ORone,\delta_1\})\redAP
\nCSP(\{R_1,R_2,\delta_0\})$.

Simplifying the chain of reductions
and using
Observation~\ref{obs:ignore} to drop~$R_1$
from the left-hand side, we get
$\nCSP(\{\ORone,R_2,\delta_0,\delta_1\})\redAP
\nCSP(\{R_1,R_2\})$.
We will now finish by showing
$\nSAT \redAP
\nCSP(\{\ORone,R_2,\delta_0,\delta_1\})$.

\noindent{\bf Case 1.\quad} Using Corollary~\ref{corI},
suppose that $t$ and $t'$ are tuples in
$R_2$ but the tuple $t\wedge t'$
(in which the operator $\wedge$ is applied componentwise)
is not in $R_2$.
We will show that $\{\ORone,R_2,\delta_0,\delta_1\}$
implements one of $\OR$ and $\XOR=\{(0,1),(1,0)\}$.
Let $k$ be the arity of~$R_2$.
As in the implementations of Creignou et al.~\cite{cks},
define
$r_i$ to be $u$ if $t_i=t'_i=0$
or $x$ if $t_i=0,t'_i=1$ or $y$ if $t_i=1,t'_i=0$, or $v$ if $t_i=t'_i=1$.
Let $R'$ be the relation implemented by
$R'(x,y)=R_2(r_1,\ldots,r_k)\wedge
\delta_0(u)\wedge
\delta_1(v)$.
Note that both $x$ and $y$ appear as arguments of~$R'$
since $t\neq t\wedge t'$ and $t'\neq t\wedge t'$.
If $t\vee t'$ is
in~$R_2$ 
then $R'(x,y)$
implements $\ORzero(x,y)$, so
we are finished.
Otherwise $R'=\XOR$
(which we now assume).

Using Observation~\ref{obs:implement} and~\ref{obs:ignore},
we have
$$\nCSP(\{\ORone,\XOR\})\redAP
\nCSP(\{R_1,R_2\}).$$ We will finish by showing that
$\{\ORone,\XOR\}$ implements $\NAND$.
(The result then follows by Lemma~\ref{lem:SATtwo}
and Observation~\ref{obs:implement}.)

The implementation is given by
$\NAND(x,z)=\Implies(x,y) \wedge \XOR(y,z)$.

\noindent{\bf Case 2.\quad} Otherwise, by
Corollary~\ref{corI},
there are
$t$ and $t'$
in $R_2$ such that $t\vee t'$ is not in $R_2$.
This case is dual to Case~1.
\end{proof}

We can now prove the main theorem.

\noindent{\bf Theorem~\ref{thm:main}.}\quad{\sl
Let $\Gamma$ be a   constraint language with
domain $\{0,1\}$.
If every relation in $\Gamma$ is affine
then $\nCSP(\Gamma)$ is in \FP.
Otherwise if every relation in $\Gamma$ is
in \IMtwo then $\nCSP(\Gamma)\equalAP\nBIS$.
Otherwise $\nCSP(\Gamma)\equalAP\nSAT$.
}

\begin{proof}

First, suppose that every relation in $\Gamma$ is affine.
In this case, the number of satisfying assignments of an instance~$I$
of $\nCSP(\Gamma)$ is the number of solutions to a system of linear equations
over GF(2). This can be computed exactly, by Gaussian elimination,
in polynomial time, as Creignou and Hermann have noted~\cite{ch}.

Next, suppose that $\Gamma$ contains a relation~$R$
that is not affine, but
every relation in $\Gamma$ is
in \IMtwo.
By Lemma~\ref{lem:downsets},
$\nCSP(\Gamma)\redAP \nBIS$

To see that $\nBIS \redAP \nCSP(\Gamma)$,
apply Lemma~\ref{otherpin}.
Then we know that
either $\nCSP(\Gamma\cup \{\delta_0\}) \redAP
\nCSP(\Gamma)$
or
$\nCSP(\Gamma\cup \{\delta_1\}) \redAP
\nCSP(\Gamma)$ (or both).
We will show
\begin{equation}\nBIS \redAP\nCSP(\Gamma\cup \{\delta_0\})
\label{eqone}\end{equation}
and
\begin{equation}\nBIS \redAP\nCSP(\Gamma\cup \{\delta_1\})
\label{eqtwo}\end{equation}
and then we will be able to conclude
$\nBIS \redAP \nCSP(\Gamma)$.
The proofs of Equations~(\ref{eqone}) and~(\ref{eqtwo}) are similar, so
we just prove~(\ref{eqone}).
By Lemma~\ref{lemsim}, $\Gamma\cup \{\delta_0\}$ implements one of
$\OR$, $\Implies$, and $\NAND$.
So by Observation~\ref{obs:implement}
we have (at least) one of the following.
\begin{enumerate}
\item $\nCSP(\Gamma \cup \{\delta_0,\OR\}) \redAP\nCSP(\Gamma\cup \{\delta_0\})$
\item $\nCSP(\Gamma \cup \{\delta_0,\Implies\}) \redAP\nCSP(\Gamma\cup \{\delta_0\})$
\item $\nCSP(\Gamma \cup \{\delta_0,\NAND\}) \redAP\nCSP(\Gamma\cup \{\delta_0\})$
\end{enumerate}
Equation~(\ref{eqone}) follows from the
combination of
Lemma~\ref{lem:BISone} and (ii) using Observation~\ref{obs:ignore}.
Also, since $\nBIS \redAP \nSAT$ (see \cite{bis}),
Equation~(\ref{eqone}) follows from the
combination of
Lemma~\ref{lem:SATzero} and (i) using Observation~\ref{obs:ignore}.
Similarly, it follows from the
combination of
Lemma~\ref{lem:SATtwo} and (iii) using Observation~\ref{obs:ignore}.

Finally,
suppose that $\Gamma$ contains a relation~$R_1$ that is not affine
and a relation~$R_2$ that is not
in \IMtwo. ($R_1$ and $R_2$ might possibly be the same relation.)
The fact that $\nCSP(\Gamma)\redAP\nSAT$ follows from
Observation~\ref{obs:SATeasy}
and the fact that $\nSAT \redAP \nCSP(\Gamma)$ follows from
Lemma~\ref{lem:main} and Observation~\ref{obs:ignore}.
\end{proof}

\section*{Appendix:
The implementations of Creignou, Khanna and Sudan}

In order to make our paper self-contained, we give the
details of the implementations of Creignou, Khanna and
Sudan that we use.
In particular, we provide the proofs for Lemmas~\ref{implementone},
\ref{cases} and \ref{lemsim}.
(These proofs can be found in \cite{cks}.)

We start with the construction for Lemma~\ref{implementone}.
Suppose $R\in \Gamma$ is not complement-closed.
Choose $(s_1,\ldots,s_k)$ in~$R$
such that
$(s_1 \oplus 1,\ldots, s_k\oplus 1)$ is not in $R$.
Now consider the relation~$R'$
implemented by $R'(x,y)=R(r_1,\ldots,r_k)$
where
 $r_i=x$ if $s_i=1$ and   $r_i=y$ otherwise.
In the first case,
$R'$ is the relation $\{(0,0),(1,1),(1,0)\}$.
In the second case,
$R'=\{(1,0)\}$ so $R'$ gives an implementation
of both~$\delta_1$ and~$\delta_0$.
The construction for the third and fourth cases are
the trivial implementations
$\delta_0(x)=R(x,\ldots,x)$
and $\delta_1(x)=R(x,\ldots,x)$.

We now give the construction for Lemma~\ref{cases}.
If $R$ excludes exactly one of $(0,1,0)$ and $(1,1,1)$
then $R(x,y,x)$ implements  $\Implies(y,x)$
or $\NAND(x,y)$ (depending on which is excluded).
Similarly, if $R$ excludes exactly one of
$(1,0,0)$ and~$(1,1,1)$ then
$R(x,y,y)$ implements $\Implies(x,y)$
or~$\NAND(x,y)$.
If both $(0,1,0)$ and $(1,0,0)$ are in~$R$
then
$f_R(x,y,z) \wedge \delta_0(z)$
implements $f_{\NAND}(x,y)$.
If $(0,1,0)$, $(1,1,1)$ and $(1,0,0)$ are excluded from~$R$
and so is~$(0,0,1)$
then
$R(x,y,z)$ implements $\ORtwo(x,y)$.
Finally, if $(0,1,0)$, $(1,1,1)$ and $(1,0,0)$ are excluded
but $(0,0,1)$ is in $R$
then
$R(x,y,z)\wedge \delta_0(x)$
implements $\ORone(y,z)$.

Finally, we give the construction for Lemma~\ref{lemsim}.
We will show that $\{R,\delta_0\}$ implements one of the named
relations. A similar argument shows that $\{R,\delta_1\}$ does.
Let $k$ be the arity of~$R$.

First, suppose that $R$ is $0$-valid.
Using part~(iii) of Lemma~\ref{affine:facts},
let $s$ and $s'$ be tuples in~$R$ such that
$s \xor s'$ is not in~$R$.
Let $r_i=w$ if $s_i=s'_i=0$.
Let $r_i=x$ if $s_i=0,s'_i=1$.
Let $r_i=y$ if $s_i=1,s'_i=0$.
Let $r_i=z$ if $s_i=s'_i=1$.
Now we know that at least one of $x$ and $y$ occurs as an $r_i$, since $s\neq s'$.
Let $R'$ be the relation implemented by $R(r_1,\ldots,r_k)\wedge \delta_0(w)$.
There are a few cases to consider.
If $x$ occurs as an argument to $R$ but $y$ does not then $z$ occurs since $s\neq 0$.
Thus, the relation $R'(x,z)$ is $\ORone$.
(Technically, this is a ternary relation in variables $x$, $y$
and $z$, but it can be viewed as a binary relation since $y$
does not appear.)
The situation is similar if $y$ occurs as an argument to~$R$ but $x$ does not.
If both $x$ and $y$ occur as arguments but $z$ does not then
the relation $R'(x,y)$ is $\ORtwo$.
Otherwise, $x$, $y$ and $z$ all occur as arguments.
Furthermore, since $R$ is 0-valid, lemma~\ref{cases} applies to
the relation given by $R'(x,y,z)$.

Second (and finally), suppose that $R$ is not 0-valid.
Note that $\{R,\delta_0\}$ can implement $\delta_1$. To see this,
let $s$ be a tuple in~$R$.
Let $r_i=x$ if $s_i=1$ and let $r_i=y$ otherwise.
Then
$\delta_1(x)$ is implemented by $R(r_1,\ldots,r_k)\wedge \delta_0(y)$.
Now consider two sub-cases.

For the first sub-case, suppose that
for any two tuples, $t$ and $t'$, in~$R$,
the tuple
$t\wedge t'$,where $\wedge$ is applied componentwise,
is also in~$R$.
Let $s$ be the intersection of all tuples in~$R$.
Then $s\in R$. By Part~(ii) of Lemma~\ref{affine:facts},
there are two tuples $s'$ and $s''$ in $R$
such that $s \xor s' \xor s''$ is not in $R$.
Let $r_i=u$ if $s_i=s'_i=s''_i=0$.
Let $r_i=x$ if $s_i=0,s'_i=0,s''_{i}=1$.
Let $r_i=y$ if $s_i=0,s'_i=1,s''_{i}=0$.
Let $r_i=z$ if $s_i=0,s'_i=1,s''_{i}=1$.
Let $r_i=v$ if $s_i=s'_i=s''_i=1$.
Let $R'$ be the relation implemented by
$R(r_1,\ldots,r_k)\wedge
\delta_0(u)\wedge
\delta_1(v)$.
If $y$ does not occur as an argument of $R'$ then
$R'(x,z)$
implements $\ORone$.
Similarly, if $x$ does not occur as an argument of $R'$ then
$R'(y,z)$
implements $\ORone$.
If $z$ does not occur as an argument of $R'$ then
$R'(x,y)$ implements  $\ORtwo$.
So we assume that $x$, $y$ and $z$ occur as argments.
Then apply Lemma~\ref{cases} to~$R'(x,y,z)$.

For the final subcase, suppose that there are tuples
$t$
and $t'$ in~$R$ such that $t \wedge t'$ is not in~$R$.
Define $r_i$ to be $u$ if $t_i=t'_i=0$
or $x$ if $t_i=0,t'_i=1$ or $y$ if $t_i=1,t'_i=0$, or $v$ if $t_i=t'_i=1$.
Let $R'$ be the relation implemented by
$R'(x,y)=R(r_1,\ldots,r_k)\wedge
\delta_0(u)\wedge
\delta_1(v)$.
If $t\vee  t'$ is in~$R$ then $R'(x,y)$  implements $\ORzero(x,y)$, so
we are finished.
Otherwise $R'=
\{(0,1),(1,0)\}$ (which we now assume).

Now using Part~(i) of Lemma~\ref{affine:facts},
let $s$, $s'$ and $s''$ be tuples in~$R$ so that
$s \xor s' \xor s''$ is not in~$R$.
Define $r_i$ as follows.

\medskip
\begin{tabular}{ccc|c}
$s_i$ & $s'_i$ & $s''_i$ & $r_i$\\  \hline
0 & 0 & 0 & $u$\\
0 & 0 & 1 & $x$\\
0 & 1 & 0 & $y$\\
0 & 1 & 1 & $z$\\
1 & 0 & 0 & ${z'}$\\
1 & 0 & 1 & ${y'}$\\
1 & 1 & 0 & ${x'}$\\
1 & 1 & 1 & $u'$\\
\end{tabular}

\medskip

Let $R''$ be the relation implemented by
$$R(r_1,\ldots,r_k)\wedge
\delta_0(u)\wedge
R'(u,u')\wedge
R'(x,x')\wedge
R'(y,y')\wedge
R'(z,z').$$
By writing $x'=\overline{x}$, $y'=\overline{y}$ and $z'=\overline{z}$,
we can think of $R''$ as a function of $x$, $y$ and $z$.
If $x$ does not occur as an argument
then $R''(y,z)$ implements $\ORone(y,z)$.
Similarly, we can assume that $y$ and $z$ occur as arguments.
Now consider the relation $R''(x,y,z)$.  We know that $(0,0,0),(0,1,1),
(1,0,1)\in R''$, since $s,s',s''\in R$.
Also $(1,1,0)\notin R''$ since $s\oplus s'\oplus s''\notin R$.
Then apply Lemma~\ref{cases} to~$R''$.

\end{document}